\documentclass[12pt]{article}
\usepackage{amsmath,amssymb,amstext,dsfont,fancyvrb,float,fontenc,graphicx,caption,subcaption,theorem,hyperref}
\usepackage[utf8]{inputenc}

\usepackage[letterpaper]{geometry}
\newlength{\BiblioSpacing}
\renewenvironment{thebibliography}[1]{
\begin{oldthebibliography}{#1}
\setlength{\parskip}{\BiblioSpacing}
\setlength{\itemsep}{\BiblioSpacing}
}
{
\end{oldthebibliography}
}

\usepackage[strict]{changepage}
\def\abstractname{Abstract -}   
\def\abstract{\begin{adjustwidth}{1cm}{1cm} \par    \footnotesize \noindent {\bf \abstractname} 
\def\endabstract{ \end{adjustwidth} \smallskip }}

{\theorembodyfont{\itshape}\newtheorem{theorem}{Theorem}[section]}
{\theorembodyfont{\itshape}\newtheorem{proposition}[theorem]{Proposition}}
{\theorembodyfont{\itshape}\newtheorem{definition}[theorem]{Definition}}
{\theorembodyfont{\itshape}}
{\theorembodyfont{\itshape}\newtheorem{corollary}[theorem]{Corollary}}
{\theorembodyfont{\rm}}
{\theorembodyfont{\rm}}
{\theorembodyfont{\rm}}
{\theorembodyfont{\rm }\newtheorem{example}[theorem]{Example}}

\title{\Large\bf Bloc Voting on Single Peaked Preferences}
\author{\sc Ariel Calver, Serena Pallan, Alice (Seoyoung) Park and Jennifer Wilson}

\begin{document}
\setcounter{page}{1}
\maketitle

\vskip 1.5em

\begin{abstract}
We analyze the winning coalitions that arise under Bloc voting when voters preferences are single-peaked. For small numbers of candidates  and numbers of winners, we  determine conditions under which candidates in winning coalitions are adjacent. We also analyze the results of pairwise contests between winning and losing candidates and assess when the winning coalitions satisfy several proposed extensions of the Condorcet  criterion to multiwinner voting methods. Finally, we use Monte Carlo simulations to investigate how frequently these coalitions arise under different assumptions about voter behavior. 
\end{abstract}
 

\begin{MSC}
91B12; 91B14
\end{MSC}

\section{Introduction}
Bloc voting is widely used in situations when several options or several individuals  must be selected. Commonly known as ``vote-for-two'' or ``vote-for-three,'' Bloc voting is used to create movie short-lists and multi-slate candidates in local elections. Its primary advantage is its simplicity: individuals vote for as many items or candidates as need to be selected. Bloc voting is less used in the political context where other multi-winner voting methods, such as standard transferable voting (STV) have been adopted to elect multiple candidates in local elections in Scotland, Australia and elsewhere (\cite{MG24}). 

Bloc voting has disadvantages: it produces results that are less ``proportional'' than STV, and it is more susceptible to undesirable outcomes such as the spoiler effect (\cite{MW23}). However, many of its  properties have not been thoroughly investigated.

 In this article, we look at what kinds of winning coalitions (sets of winning candidates) arise when Bloc voting is used to select a small number of winners from a small pool of candidates. We also  analyze how these winning coalitions  behave in relation to the Condorcet winner.   The Condorcet winner is a candidate
 that  defeats every other candidate in a head-to-head contest. In single-winner voting methods, the notion of a Condorcet winner  has  played a prominent role in assessing the ``fairness'' of a voting method, and a voting method that always select the Condorcet winner when it exists is called  Condorcet-consistent or a Condorcet method.  However, there is no single way to usefully extend the idea of a Condorcet winner  to the multi-winner context. Several different notions have been proposed, which we explore using Bloc voting.

We assume voters used ranked ballots, in which they indicate their preferences for each candidate. Having complete preference information, while not necessary to run a Bloc election, allows us to determine the result of all possible pairwise contests.  We also confine our discussion to single-peaked preferences, in which voters' preferences are consistent with the arrangement of candidates  along a political spectrum (left to right). When voters' preferences are single-peaked, there is always a Condorcet winner. This makes single-peaked preferences a useful setting in which to assess how the winning candidates in a Bloc voting election fair in  head-to-head contests with non-winning candidates. It also allows us to determine when candidates are subject to the {\it center squeeze} phenomenon in which a centrist candidate loses in favor or more extreme candidates.

We determine conditions under which the winning coalitions  are adjacent. We also classify which winning coalitions satisfy different Condorcet criteria and local stability conditions. We provide a complete characterization when there are four, five, or six candidates and a partial analysis for a larger number of candidates.  In the last part of the paper, we use Monte Carlo simulations to investigate the probability that winning coalitions under Bloc voting satisfy the different Condorcet criteria based on three different models of voter preferences.

There is a growing literature on the behavior of Bloc voting and other multi-winner voting methods.  The authors in \cite{EF17a} compare several voting methods under different models of voter behavior. The properties of these methods are discussed in \cite{EF17a}; \cite{SF19} contains an axiomatic characterization of scoring rules, a class that includes Bloc voting. Extensions of the Condorcet principle to multi-winner voting methods are discussed in \cite{HE17}. Additional references appear throughout the paper.

The remainder of the discussion proceeds as follows. In Section 2, we introduce Bloc voting, formalize the single-peaked preference model, and define three different generalizations of the Condorcet criterion to the multi-winner setting. In Section 3, we analyze the winning sets vis-a-vis head-to-head elections when there are four and five candidates. In  Section 4 we prove several results for an  arbitrary numbers of candidates and apply the results to   six and seven-candidate elections. Section 5 describes the results of the Monte Carlo simulations.  We conclude in Section 6.

\section{Bloc Voting, Single-Peaked Preferences and Condorcet Sets}\label{sec_prelims}

 There are several  ways to collect information about voters' preferences over a set of candidates. Ballots may ask voters to:  ``vote for one'' or ``vote for k'' candidates;   vote for as many candidates as they approve of (so-called approval ballots); award points to candidates; or  rank some or all of the candidates. In this paper, we focus on the latter. We assume that each voter has a strict linear preference over all the candidates (no ties), so they can rank order them in order from most-preferred to least-preferred. This is a strong assumption; in actual elections, voters are often allowed to only rank  only a fixed number of candidates (which may be less than the number running). Additionally, voters often choose to rank fewer candidates than they are allowed, a phenomenon known as ballot truncation.
 
 We assume there are $N$ voters and $m$ candidates, and that $k$ winners are to be selected, where $k$ is a fixed integer $1 \le k < m$.  A multi-winner voting method is a function that inputs the set of voter preferences (called a preference profile)   and outputs a set of $k$ winners. Because of the possibility of ties, voting methods are often considered multi-valued. To avoid this complication, we assume that $N$ is odd and focus primarily on situations without ties.
 
Under  Bloc voting, each voter casts votes for the top $k$ candidates in their ranking.  The $k$ candidates with the greatest number of votes are  declared winners. (If $k=1$, the winner is the candidate with the most first-place votes,  known as the plurality winner.)  
We illustrate with the following example.

\begin{example}
\label{ex_first} \textbf{Bloc Voting} Suppose there are five candidates $A, B, C, D$ and $E$ and 145 voters whose    ballots are as shown in  Table \ref{ex_initial}. The  top row indicates the number of voters with the given ranking. For example, there are 50 voters rank the candidates $A \succ B \succ C \succ D \succ E.$  

\begin{table} [h]
\begin{center}
$$\begin{array}{ccccc}
50 & 40 & 20 & 30 & 5  \\
\hline
A & B & C & D & E \\
B & E & B & B & D \\
C & C & D & E & B \\
D & D & E & C & C \\
E & A & A & A & A 
\end{array}$$
\caption{A voter profile with five candidates.}
\end{center}
\end{table}\label{ex_initial}

If $k=1$  each individual votes only for their top choice. Candidate $A$, with 50 votes, is the winner.
If $k=2$, $A$ receives votes only from the 50 individuals in column 1. $B$  receives votes from individuals in columns 1 through 4, for a total of 140 votes. Similarly, $C$ receives 20 votes, $D$ receives 35 votes, and $E$ receives 45 votes. Thus the winning set is $AB$.
If $k=3$,  $B$ receives an additional 5 votes for a total of 145 votes, $C$ receives 110 votes, $D$ receives 55 votes, and $E$ receives 75 votes, so the winning set is $BCE.$ If $k=4$, the winning set is $BCDE.$ \end{example}

Example \ref{ex_first}  demonstrates an interesting phenomenon that a candidate who is a winner when selecting 
$\ell$  candidates may not be a winner when selecting $\ell+1$ candidates.  In this example, $A$  wins when
$k=1$ or $k=2$ but loses when $k=3$ or $k=4$.
Because of this possibility, Bloc voting is said to violate \emph{committee monotonicity}.

Whether committee monotonicity is an important property   depends on the context. If the goal of the election is to select the ``top k'' candidates to interview for a job, then it is perhaps disturbing that a candidate may fall from a shortlist if the size of the list increases. On the other hand, if the goal of the election is to select a set of political representatives, it may be that as $k$ increases, centrist candidates, for instance, are replaced with candidates representing left and right wing segments of the electorate.

To provide a context for Bloc voting, and to introduce the notion of head-to-head elections, we define a second  voting method, Copland's method. Under Copland's method, candidates receive one point for each head-to-head contest they win against another candidate (and half a point if the result is a tie). The candidates are then ranked based on the total number of points they receive. Copland's method is most often used to find a single winner, but it can also be used to select $k$ winners by picking the $k$ candidates with the highest Copland scores.

\begin{example} \label{ex_cop} \textbf{Copland's Method} 
Suppose the preference profile is as in Table \ref{ex_initial}. In a head-to-head contest between $A$ and $B$, each voter selects their preferred candidate. Candidate $A$ receives 50 votes from the first column while $B$ receives 95 votes from the remaining columns.  Thus $B$ wins the head-to-head contest and receives one point.  The results of the remaining head-to-head contests are as follows. 

\begin{center}
\begin{tabular}{c|c|c}
Head-to-head & Vote Total & Winner \\
\hline
$A$ vs. $B$ & $50$ to $95$ & $B$  \\
$A$ vs. $C$ & $50$ to $95$ & $C$ \\
$A$ vs. $D$ & $50$ to $95$ & $D$ \\
$A$ vs. $E$ & $50$ to $95$ & $E$ \\
$B$ vs. $C$ & $125$ to $20$ & $B$ \\
$B$ vs. $D$ & $110$ to $35$ & $B$ \\
$B$ vs. $E$ & $140$ to $5$ & $B$ \\
$C$ vs. $D$ & $110$ to $35$ & $C$ \\
$C$ vs. $E$ & $70$ to $75$ & $E$  \\
$D$ vs. $E$ & $105$ to $40$ & $D$ \\
\end{tabular}
\end{center}

The Copland scores for $A, B, C, D$ and $E$ are  $0, 4, 2, 2$ and $2$ respectively, resulting in an ordering of $B \succ C \approx E \approx D \succ A.$
If $k=1$, the winner is $B$; if $k=2$, there is a tie, resulting in three winning sets: $BC$, $BD$ and $BE$. If  $k=3,$ there is another tie and the winning sets are:$BCD$, $BCE$, $BDE$ and $CDE.$ If $k=4$, the winning set is $BCDE.$ 
\end{example}

 When used as a single-winner voting method, Copland's method is  Condorcet-consistent, as the Condorcet winner, if it exists, will always have the highest score. (This is the case for candidate $B$ in Example \ref{ex_cop}.) When extended to the multi-winner context, Copland's method   clearly satisfies committee monotonicity. 
 
\subsection{Single-Peaked Preferences} 
Single-peaked preferences were first introduced  in \cite{B48}. Under a single-peaked model, candidates are linearly ordered, and voters' preferences are consistent with this ordering. Usually, the relative position of the candidates is  interpreted as placement along a one-dimensional left-right political spectrum.  Voters may have differing perceptions about the   distances between candidates. However they all agree on the ordering. 

This assumption restricts the  candidate rankings that are possible. For instance, suppose four candidates are ordered left-to-right as in Figure \ref{fig_spectrum}. A voter who prefers $D$ must rank $C$ second, followed by $B$ and $A$ (denoted $D \succ C \succ B \succ A$). A voter who prefers $C$, may rank either $B$ or $D$ second.

\begin{figure}[!htb]
\begin{center}
        \includegraphics[width=0.5\linewidth]{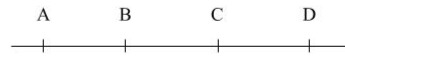}
\end{center}
 \caption{\label{fig_spectrum} In this model, A represents the most left-wing candidate, B and C are moderate candidates, and D is the most right-wing candidate.}
 \end{figure}

A consequence of single-peaked preferences is that no voter's preferences can ``skip'' a candidate. In Figure \ref{fig_spectrum}, no voter may have preferences $B \succ D \succ C \succ A$.  
The profile in Example \ref{ex_first} cannot represent single-peaked preferences since there are voters whose preferences include $C \succ D \succ E$ (first column), $E \succ C \succ D$ (second column), and $D \succ E \succ C$  (third column). And it is not possible to arrange $C$, $D$ and $E$ in such a way that all of these preferences avoid such skips.

The term ``single-peaked'' refers to the fact that if each voter's preferences are represented on a graph with the candidates arranged left-to-right on the horizontal axis, and voter preferences represented on the vertical axis, then each voter's preferences have a single peak. 

\subsection{Head-to-Head Contests on Single-Peaked Preferences}
A distinguishing feature of single-peaked preferences is that head-to-head contests produce a linear ordering of candidates. This means there exists a Condorcet winner who beats every other candidate, a candidate who beats every other candidate except the Condorcet winner, and so on. This result, known as the Median Voter Theorem \cite{B48}, follows because the median candidate (assuming an odd number of voters) must be the Condorcet winner. In a head-to-head contest with a candidate to the left, the median candidate captures the right half of voters plus voters between the two candidates. The same holds for contests with candidates to the right.

If the Condorcet winner is eliminated and ballots adjusted by moving lower-ranked candidates up, the new median candidate wins all head-to-head contests against remaining candidates. Repeating this process until no candidates remain produces a linear ordering, as illustrated in the following example.

\begin{example}Consider the profile of 211 voters shown  in Table \ref{single_peaked} (Left). If the first-place votes are lined up as shown below,

$$\underbrace{A, A, \ldots A}_{50\rm\ times}, \underbrace{B, B, \ldots, B}_{51 \rm\ times},  \underbrace{C, C, \ldots, C}_{90 \rm\ times}, \underbrace{D, D, \ldots, D}_{20 \rm\ times},$$
we  see that the median candidate (in the 106th place) is $C$. Removing $C$ results in the adjusted profile shown in Table \ref{single_peaked} (Right) where the new median candidate is $B$. Another elimination produces the next median candidate, which is $A$.  
Thus, $C$ beats every candidate in a head-to-head contest; $B$ beats candidates $A$ and $D$ in head-to-head contests and $A$ beats $D$ in a head-to-head contest, resulting in Copland scores of $3, 2, 1$ and $0$ respectively.
\end{example}
\begin{table}[h]
$$\begin{array}{cccc}
50 & 51 & 90 & 20   \\
\hline
A & B & C & D  \\
B & A & B & C  \\
C & C & A & B   \\
D & D & D & A   \\
\end{array}
\hskip 1cm
\begin{array}{ccc}
50 & 141 & 20    \\
\hline
A & B & D  \\
B & A & B  \\
D & D & A   \\
\end{array}$$
\caption{A preference profile and adjusted profile representing singe-peaked preferences.}\label{single_peaked}
\end{table}

We call the ordering induced by the head-to-head rankings (or, equivalently, the ranking induced by the Copland scores), as the {\it Condorcet ranking}.  Because of this ordering,  he Copland scores if there are $m$ candidates (assuming an odd number of voters) will always be, arranged in decreasing order, $m-1, m-2, \ldots, 0.$

If voter preferences are  not single-peaked, such an ordering may not be possible. This is the case in Example 1, where $C$ beats $D$ in a head-to-head contest, $D$ beats $E$ in a head-to-head contest and $E$ beats $C$ in a head-to-head contest, resulting in what is known as a Condorcet cycle.

\subsection{Condorcet Extensions to Multi-Winner Elections}
There have been several  generalizations of the Condorcet winner to sets of winners proposed for multi-winner elections. The first notion was suggested independently in  \cite{G03} and  \cite{R03}. Following the language in \cite{HE17}, we will refer to these sets as Gehrlein-stable.

\begin{definition}
A set of candidates $W$ is Gehrlein-stable if, any any head-to-head contest between a candidate in $W$ and a candidate not in $W$, the candidate in $W$ wins. 
\end{definition}

Every preference profile has a trivial Gehrlein-stable set: the set of all candidates. However, a Gehrlein-stable set of a fixed size $k$ may not exist.

Elkind et al. \cite{HE17} define a weaker notion:
\begin{definition}
A set of candidates $W$ is a Condorcet set or simply Condorcet, if for every candidate not in $W$, there exists a candidate in $W$ who wins in a head-to-head contest against the candidate not in $W$. 
\end{definition}

Every Gehrlein-stable set is a Condorcet set, but the reverse is not true. In Example 1, the Bloc winners when $k=2$, $AB$, form a Condorcet set since a majority of voters prefer $B$ (the Condorcet winner) to any every other candidate. However, $AB$ is not Gehrlein-stable since 95 voters prefer $C$, $D$, or $E$ to $A$.

For single-peaked preferences, these definitions simplify as follows.
\begin{proposition}
Under single-peaked preferences with an odd number of voters, a set is Condorcet  if and only if it contains the Condorcet winner. A set of size $k$ is Gehrlein-stable if and only if it consists of the top $k$ candidates in the Condorcet ranking.
\end{proposition}

Under single-peaked preferences, the candidates in a Gehrlein-stable set must be adjacent.

\begin{proposition}\label{Geh_adj}
Under single-peaked preferences, every Gehrlein-stable set consists of adjacent candidates.
\end{proposition}
\begin{proof}
Suppose the candidates have Copland ranking  $C_1 \succ C_2 \succ C_3 \succ \cdots,C_n$. Without loss of generality, assume $C_2$ lies to the right of $C_1$. Suppose $C_1$ and $C_2$ are not adjacent and let $C^\prime$ be the candidate just to the left of $C_2$. Any voter who prefers $C_1$ to $C_2$ must also prefer $C^\prime$ to $C_2$ (since $C^\prime$ lies between them). Hence in a head-to-head contest, $C^\prime$ does at least as well against $C_2$ as $C_1$ does. Since $C_1$ beats $C_2$, so must $C^\prime.$ But this contradicts $C_2$ being ahead of $C^\prime$ in the ordering. Hence $C_1$ and $C_2$ must be adjacent. The same argument shows $C_2$ must be adjacent to $C_3$, and so on.
\end{proof}

Thus, any set of candidates $W$ that is non-adjacent cannot be Gehrlein-stable, implying, there is a candidate not in $W$ that beats some candidate in $W$ in a head-to-head contest, as illustrated below.

\begin{example}\label{ex_center_squeeze} Suppose there are 5 candidates $A, B, C, D$ and $E$, arranged left-to-right. Suppose there are 61 voters with preference profile as shown in Table \ref{ex_nonadj}.  \begin{table}[h]
\begin{center}
 $$ \begin{array}{cccc}
20 & 10 & 11 & 20  \\
\hline
B & C & C & D \\
A & B & D & E \\
C & A & E & C \\
D & D & B & B \\
E & E & A & A
\end{array}$$ 
\end{center}
\caption{A preference profile with a non-adjacent winning set.}
\label{ex_nonadj} 
\end{table}
Under Bloc voting with $k=2$, the winning set is are $BD.$  However, in a head-to-head contest, $C$ beats both $B$ and $D$ and is the Condorcet winner.
\end{example}

Example \ref{ex_center_squeeze} illustrates a phenomenon that  occurs frequently  in  multi-winner elections: the center squeeze, where a centrist candidate loses out  to candidates on either side of it.  The center squeeze cannot occur under the Copland method. (In Example \ref{ex_center_squeeze}, the Copland winners are $CD$.) In fact, since the Copland method selects as winners the top $k$ candidates in the Condorcet ranking, we have the following. 
\begin{corollary}
Any winning set under the Copland method is adjacent and Gehrlein-stable.
\end{corollary}

Gehrlein-stability, while compelling, does not place any restrictions on {\it who} prefers one candidate to the other. In Example \ref{ex_center_squeeze}, for instance, the  voters who prefer $C$ to $B$ are the 41 voters on the right-hand side. The voters who prefer $C$ to $D$ are the 40 voters on the left-hand side. The only voters who prefer $C$ in  head-to-head contests with both winners are the 21 voters in the center.

Elkind et al \cite{HE17}  proposed an alternative notion to Gehrelein-stability to capture situations when a ``block'' of voters  prefer a losing candidate to the winning candidates.  

\begin{definition}
A set of candidates $W$ is locally stable for quota $q$ if there exists no set of voters of size at least $q$ and non-winning candidate $C$ such that all voters in the set prefer $C$ to each candidate in $W$.  
\end{definition}

Alternatively, a winning set is not locally stable for quota $q$ if there exists a set of voters of size at least $q$ and a non-winning candidate $C$ such that all voters in the set prefer $C$ to each candidate in $W$.
The authors in  \cite{HE17} discuss the consequences of selecting different values of $q$, suggesting the Droop quota $\hat{q}= \lfloor \frac{N}{k+1}\rfloor+1$, used in single transferable vote systems. (The Droop quota ensures no more than $k$ candidates can meet or exceed $\hat{q}$  first-place votes; hence any candidate meeting this quota ``deserves'' to win.)

In this paper, we let $q=\lfloor \frac{N}{2}\rfloor+1$ for all $k$ values. Hence we define locally stable sets as those for which there is no non-winning candidate preferred by a majority of voters to each of the winners. 
Under this definition, Gehrlein-stability implies local stability, since if a set is Gehrlein-stable, no losing candidate can be majority-preferred to any winning candidate.
Example \ref{ex_center_squeeze} provides an example of a winning set $BD$ that is not Condorcet but is  locally stable since, although $C$ is the Condorcet winner, there is no  group of at least 31 voters who prefer $C$ to both $B$ and $D$. (It is not locally stable if $q$ is the Droop quota $\hat{q}= \lfloor \frac{61}{2+1}\rfloor+1= 21.$)

\section{Bloc Voting with Four or Five Candidates}\label{sec_4_and_5}
In this section, we classify the winning coalitions that can arise under Bloc voting when there are four or five candidates and voters have single-peaked preferences.  We also investigate when these coalitions are adjacent, and when they are Gehrlein-stable, Condorcet, or locally stable.

With four candidates, there are 8 possible rankings  under single-peaked preferences. If the  candidates,  arranged left to right, are labeled  $A, B, C,$ and $D$, then the preference profile can be described as  in Table \ref{4_cand}
\begin{table}[h]
\begin{center}
$$\begin{array}{ccccccccc}
x_1 &x_2 & x_3 & x_4 & x_5 & x_6 & x_7 & x_8\\
\hline
A & B& B & B& C & C & C& D\\
B&  A & C& C &B & B & D & C\\
C&  C& A& D&  A & D & B & B\\
D & D& D& A & D &  A &  A&A
\end{array}$$
\end{center}
\caption{A preference profile with four candidates.}
\label{4_cand}
\end{table}
where $x_1+ x_2+ \cdots + x_{8}=N.$ If $k=2$, the Bloc vote totals for each candidate are:
\begin{align*}
T_A &= x_1+x_2,\\
T_B & = x_1+x_2+ \cdots + x_6, \\
T_C & = x_3+x_4+ \cdots + x_{8} \mbox{ and} \\
T_D & =  x_{7}+x_{8}.\\
\end{align*}
Since $T_A \le T_B,$ if $A$ is a winner under Bloc voting then so is $B$. Likewise, $T_D\le T_C$ so if $D$ is a winner  so is $C$. Thus the only possible winning sets are $AB$, $BC$, and $CD.$ If $k=3$, a similar argument shows the only winning sets are $ABC$ and $BCD.$

With five candidates, $A, B, C, D,$ and $E$, arranged left-to-right, there are 16 possible rankings, with preference profile  as shown in Table \ref{5_cand}, where $x_1+x_2+ \cdots + x_{16}=N$. 
\begin{table}[h]\begin{center}
$$\begin{array}{cccccccccccccccc}
x_1 &x_2 & x_3 & x_4 & x_5 & x_6 & x_7 & x_8 & x_9 &x_{10} & x_{11} & x_{12} & x_{13} & x_{14} & x_{15} & x_{16}\\
\hline
A&B & B & B & B &  C &C &C &C &C &C &D& D&D&D& E\\
B&A& C&C&C&B&B&B&D&D&D&C&C&C&E&D\\
C&C&A&D&D&A&D&D&B&B&E&B&B&E&C&C\\
D&D&D&A&E&D&A&E&A&E&B&A&E&B&B&B\\
E&E&E&E&A&E&E&A&E&A&A&E&A&A&A&A
\end{array}$$
\end{center}
\caption{A preference profile with five candidates.}
\label{5_cand}
\end{table}
 If $k=2$, $T_A \le T_B$ and $T_E\le T_D$. So if $A$ (respectively $E$) is a winner, so is $B$ (respectively $D$). Thus the only possible winning sets are: $AB, BC, BD, CD,$ and $DE$. If $k=3$, $T_A \le T_B \le T_C$ and $T_E\le T_D\le T_C$, so the only possible winning sets are: $ABC, BCD,$ and $CDE$. If $k=4$, the only possible winning sets are: $ABCD$ and $BCDE.$
Thus, all winning sets consist of adjacent candidates with the exception (when $k=2$) of $BD.$

\subsection{Head-to-Head Contests}

To analyze which winning sets are Gehrlein-stable, we must determine the winner of each head-to-head contest between a winning and a non-winning candidate.  
For any two candidates $C_i$ and $C_j$, let $T_{C_iC_j}$ denote the number of voters who prefer $C_i$ to $C_j$. For $C_i$ to win over $C_j$ in a head-to-head contest requires $T_{C_iC_j}>T_{C_jC_i}$ or, equivalently (since we are assuming $N$ is odd), $T_{C_iC_j} > N/2$.

\begin{proposition}\label{Geh-stable}
For four candidates, if $k=2$ or $3$, all winning coalitions under Bloc voting are Gehrlein-stable. For five candidates, if $k=2, 3,$ or $4$, all winning coalitions under Bloc voting are Gehrlein-stable except (when $k=2$) $AB$, $BD$, and $DE.$
\end{proposition}
\begin{proof}
We prove the case for five candidates; the proof for four candidates is analogous. Suppose, first, that $k=2$ and that the winning set is $BC$. Consider a head-to-head contest between $B$ and $D.$ Since $B$ is a Bloc winner but $D$ is not,

$$T_{DB} = x_9+x_{10}+ + \cdots + x_{16} =T_D<T_B = x_1+ x_2+ \cdots + x_8 = T_{BD}.$$
So $D$ cannot win against $B$ in a head-to-head contest. Next, consider a head-to-head contest between $C$ and $D.$ Since $T_{DC} =x_{12}+ \cdots x_{16} < T_{DB},$ we have
$T_{DC}< T_{DB} < T_{BD} < T_{CD}.$ So $D$  cannot win against $C$ in a head-to-head contest either.
Since $E$ fares worse than $D$ in head-to-head contests with $B$ or $C$ ($T_{EB}<T_{DB}$ and $T_{EC}<T_{DC}$), $E$ also cannot win against either winner.

It remains to consider $A.$ In a head-to-head contest between $A$ and $C$, since $C$ is a Bloc winner but not $A$,  
$$T_{AC} =x_1 +x_2 =  T_A < T_C=x_3+ \cdots + x_{14} < x_3+ \cdots + x_{16} = T_{CA}.$$
So $A$ cannot win against $C.$ 
Additionally, $T_{AB} < T_{AC}< T_{CA} < T_{BA}$, which implies $A$ cannot win against $B$ either.
Thus the winning coalition $BC$ is Gehrlein-stable. 
By symmetry, the same argument applies to $CD.$

Now consider $k=3$. Suppose the winning set is $ABC$. Since $A$ is a Bloc winner and $D$ is not,  

$$T_{DA} = x_4+x_5+x_7+x_8+ \cdots x_{16} =T_D<T_A =x_1+\cdots + x_3+x_6=T_{AD}.$$
Thus $D$ cannot beat $A$ in a head-to-head contest, and since $D$ fares worse against both $B$ and $C$, $D$ cannot win against any winner. The same is true for $E.$ Thus $ABC$ is Gehrlein-stable. By symmetry, the same holds for $CDE.$

Finally, suppose the winning set is  $BCD.$ Since $D$  is a Bloc winner but not $A$,
 $$T_{AD}  =T_A <T_D= T_{DA},$$
so $A$ cannot win against $D.$ Since $A$ performs worse against $B$ and $C$, it cannot win against any winning candidate. By symmetry, the same is true of $E$; thus $BCD$ is also Gehrlein-stable.

\end{proof}

 \begin{proposition}\label{prop_5_2}
 Suppose $m=5$ and $k=2$. Under Bloc voting:
 \begin{enumerate}
 \item If the winning set is $AB$ then $B$ cannot be beaten in a head-to-head contest with $D$ or $E$, but there exists a profile in which $B$ can be beaten in a head-to-head contest with $C$ and in which $A$ can be beaten by all three non-winning candidates. An analogous statement is true for $DE.$
 \item If the winning set is $BD$, then $B$ cannot be beaten by $A$, and $D$ cannot be beaten by $E$. Furthermore, there exists a profile in which: (i) $E$ beats $B$ and $C$ beats both $B$ and $D$ in head-to-head contests; and another profile in which (ii) $A$ beats $D$ and $C$ beats both $B$ and $D$ in head-to-head contests.
 \end{enumerate} 

 \end{proposition}

\begin{proof} 1. 
Suppose the winning set is $AB$. As shown in the proof of Proposition \ref{Geh-stable}, if $B$ is a Bloc winner but not $D$, then
$$T_{EB}<T_{DB} =T_D< T_B = T_{BD}< T_{BE}.$$
Hence neither $D$ nor $E$ can win in a head-to-head against $B.$

Now consider the  profile in Table \ref{5cand_counter} (Left) with $251$ voters. The winners are $A$ and $B$ with $101$ and $151$ votes respectively. However, both $D$ and $E$ beat $A$   and $C$ beats both  $A$ and $B$ in head-to-head contests ($151$ to $101$ in all cases).  

\begin{table}[h]
$$\begin{array}{cccc}
101 & 50 & 50 & 50   \\
\hline
A & C & C & D \\
B & B & D & E \\
C & D & E & C \\
D & E & B & B \\
E & A & A & A  \\
\end{array}
\hskip 1cm
 \begin{array}{ccccc}
 101 & 1 & 1 & 2 & 100  \\
\hline
 B & C & C & D & D \\
 A & B& D & C & E \\
 C & D& E & E & C \\
 D & E& B & B & B\\
E   &A&  A & A & A
\end{array} $$
\caption{Preferences in which non-winning candidates beat winning candidates in head-to-head contests when the winning set is $AB$ (Left) or $BD$ (Right).}\label{5cand_counter}
\end{table}
2. Suppose the winning set is $BD$. Since $B$ is a Bloc winner but not $A$, $A$  cannot beat $B$ in a head-to-head  because 
$$T_{AB} =x_1 < T_A<T_B< x_2+ \cdots + x_{16} = T_{BA}.$$
Similarly, $E$ cannot beat $D$. Now consider the  profile in Table \ref{5cand_counter} (Right) with  $205$ voters. The Bloc winners are, $B$ and $D$ with $102$ and $103$ votes respectively. In head-to-head contests, $E$ beats $B$ (103 to 102) and $C$ beats both $B$ (104 to 101)  and $D$ (103 to 102).

In the symmetrically reversed profile, $A$ beats $D$ and $C$ beats both $B$ and $D$ in head-to-head contests.  

 \end{proof}

Proposition \ref{prop_5_2} shows that with five candidates, if the winning set is $AB$, $DE$, or $BD,$ the winning set may not be Gehrlein-stable, or even Condorcet, since the centrist candidate may be majority-preferred to either winner. 

Moreover, the counterexample in the proof of (1) demonstrates that if the winning set is $AB$ (and by symmetry $DE$), it may not be locally stable. 
This is not the case if the winning set is $BD.$

\begin{proposition}\label{BD_set} If there are 5 candidates and $k=2,$ then all winning sets under Bloc voting are locally stable except $AB$ and $DE.$ 
\end{proposition}
\begin{proof}
If the winning set is $BC$ or $CD$ then it is Gehrlein-stable and hence locally stable. If the winning committee is $BD$ it is possible for $C$ to beat beat both $B$ and $D.$ However, for a majority of voters to prefer $C$ to either winner, the quantity $x_6+ \cdots x_{11}$ must correspond to a majority of voters. But since $C$ is not a Bloc winner,
$$T_B = x_1+ \cdots x_8 > x_3+ \cdots + x_{14}=T_C.$$

Subtracting $x_3+ \cdots + x_8$ from both sides, we obtain

$$x_1+x_2> x_9+ \cdots +x_{14},$$

which implies 
$$x_1+x_2+ \cdots + x_8 > x_9+ \cdots +x_{14}.$$
Since $x_1+ x_2 + \cdots + x_8 = N-(x_9+\cdots + x_{14}),$  this implies 
$$x_9+ \cdots x_{14}<N/2$$
and hence $x_6+ \cdots x_{11}  <N/2.$  Thus there is no majority who prefers $C$ to both $B$ and $D.$
\end{proof}

In the next section, we generalize many of these observations for more than five candidates. 

\section{More than five Candidates}

When there are more than five candidates, it is cumbersome to identify each possible ranking as $x_1, x_2$, etc, and so we adapt a more general notation.
Suppose there are $m$ candidates, arranged left-to-right $C_1, C_2, \ldots, C_m.$ Let $S_i$ be the set of voters who rank $C_i$ among their top $k$ and let $\vert S_i \vert$ denote the number of these voters. Let  $k=1, 2, \ldots, m-1$ be the number of winners required and  recall that there are $N$ voters where $N$ is odd.

\subsection{Winning Sets}
As with four and five candidates, not all sets of $k$ candidates can be winning sets under Bloc voting. 

\begin{proposition}
If voters preferences are single-peaked and $C_i$ is a winning candidate under Bloc voting for some $i \le k$, then $C_{i+1}, C_{i+2}, \ldots, C_k$ are also winning candidates under Bloc voting. If $C_i$ is a winning candidate for some $i \ge m-k+1$, then $C_{m-k+1}, C_{m-k+2},\ldots, C_{i+1}$ are also winning candidates.
\end{proposition}

\begin{proof}
The first statement follows since if $i \le k$ then $S_i \subset S_k$, which implies $\vert S_i \vert \le \vert S_k \vert.$ The second statement follows analogously.
\end{proof}

Thus, as with four and five candidates, if a winning set contains a candidate within $k$ of the left or right extreme, then it must contain adjacent candidates that are more central. For example, if there are six candidates $A, B, C, D, E$ and $F$ then if a winning set contains $B$, it must contain $C$, and if a winning set contains $A$, it must contain $A, B$ and $C$.

Next we consider how the winning candidates must be arranged when $k$ is relatively large compared to $m$. 

\begin{proposition}\label{adj}
If  $k \ge \lceil m/2 \rceil$ then the winning coalition is adjacent. 
\end{proposition}
\begin{proof}
We prove the case when $m$ is even.
Let $k \ge m/2$ and suppose the winning set is not adjacent. Let  $C_\ell$ be a non-winning candidate such that there are winning candidates to both left and right of $C_{\ell}$. Suppose,
 without loss of generality, that $\ell \le m/2$.  Let $C_i$ be the   winning candidate that is closest to $C_{\ell}$ and to their left. Since $i \le \ell \le k$, any voter who includes $C_i$ among their top $k$ candidates must also include $C_{\ell}$ among their top $k$ candidates.  Hence  $S_i \subset S_{\ell}$, which implies $\vert S_i \vert \le \vert S_{\ell} \vert.$ But this is impossible since $C_i$ is a winner and  $C_\ell$ is not. Thus, the  winning set must be adjacent.  
\end{proof}

\subsection{Gehrlein-stability}
To determine which winning sets are Gehrlein-stable, we must analyze the result of all head-to-head contests between winning and non-winning candidates. For each pair of candidates $C_i$ and $C_j$, let  $S_{i,j}$ be the set of voters who prefer $C_i$ to $C_j$ and let  $\vert S_{i,j}\vert$ be the number of those voters. To prove that $C_i$  wins in a head-to-head contest with $C_j$, it is necessary to show that $\vert S_{i,j}\vert> \vert S_{j,i}\vert$ or, equivalently, that $\vert S_{i,j}\vert> N/2$.

Before analyzing these contests, we note that the linear ordering of the candidates implies certain relationships among the $S_{i, j}$. For instance, suppose a voter prefers $C_4$ to $C_6$. Then they must also prefer $C_4$ to $C_7$, since otherwise their preferences would be $C_7 \succ C_4 \succ C_6$, which is impossible. Thus $S_{4,6} \subset S_{4, 7}$. We generalize this relationship in the following observation.

\textbf{Observation:} 
 Suppose   $i<j<k$ so that $C_i$,  $C_j$, and $C_k$ are arranged  left to right. Then any voter  who prefers $C_i$ to $C_j$ must prefer $C_i$ to $C_k.$ Hence $S_{i,j} \subset S_{i,k}$ which, taking complements, implies that $S_{k,i} \subset S_{j,i}$
  Likewise, any voter who prefers $C_k$ to $C_j$ must prefer $C_k$ to $C_i$. Thus $S_{k,j} \subset S_{k,i}$ and  $S_{i,k} \subset S_{j,k}$.

\begin{proposition}\label{k_away}
1. Suppose $C_i$ is a non-winning candidate for some $i \le k$ and $C_{j}$ is a winning candidate with $j \ge i+k$. Then   $C_i$ cannot win in a head-to-head contest with any of the candidates $C_{i+1}, C_{i+2}, \ldots, C_{i+k}$. Additionally, no candidate to the left of $C_i$ can win in a head-to-head contest with any of the candidates  $C_{i+1}, C_{i+2}, \ldots, C_{i+k}$.

2. Suppose $C_i$ is  a non-winning candidate for some $i \ge m-k+1$ and $C_{j}$ is a winning candidate such that $j \le i-k$. Then $C_i$ cannot win in a head-to-head contest with any of the candidates  $C_{i-k}, C_{i-k+1}, \ldots, C_{i-1}$.
Additionally, no candidate to the right of $C_i$ can win in a head-to-head contest with any of the candidates  $C_{i-k}, C_{i-k+1}, \ldots, C_{i-1}$.
\end{proposition}

\begin{proof}
We prove  statement (1). Consider the set $S_{i, i+k}$. This  can be partitioned into two subsets $T_1 \cup T_2$ where $T_1$ consists of all voters who  first-rank $C_{i}$ or any candidate to their left, and $T_2$ consists of all voters who  first-rank candidates between $C_{i}$ and $C_{i+k}$ and who prefer $C_i$ to $C_{i+k}$. Since  $i \le k$, any voter in $T_1$  must include $C_{i}$ among their top $k$. Also, since $C_{i+k}$ is $k$ away from $C_i$,  any voter in $T_2$ must also include $C_{i}$ among their top $k$. Thus,  $S_{i, i+k} \subset S_{i},$ which implies, 
\begin{equation}\label{eq_1}
\vert S_{i, i+k} \vert \le \vert S_{i}  \vert.\end{equation}
On the other hand, since $C_j$ is at least $k$ away from $C_i$, any voter who includes $C_j$ among their top $k$ must prefer $C_j$ to $C_i$, and hence lies in $S_{j, i}$. Additionally, by the Observation, since $C_{i+k}$ lies between $C_i$ and $C_j$,  any voter who prefers $C_j$ to $C_i$  must prefer $C_{i+k}$ to $C_i$. Hence $S_{j} \subset S_{j, i} \subset S_{i+k, i},$ which implies 
\begin{equation}\label{eq_2}
\vert S_j \vert \le \vert S_{j,i} \vert \le \vert S_{i+k,i} \vert.\end{equation}
Finally,   since $C_i$ is a non-winning candidate and $C_j$ is a winning candidate, we must have $\vert S_i\vert < \vert S_j \vert.$
Putting (\ref{eq_1}) and (\ref{eq_2})  together yields 

$$\vert S_{i, i+1} \vert  \le \vert S_i \vert < \vert S_j \vert \le \vert S_{i+k,i} \vert$$
and hence $\vert S_{i, i+k} \vert< N/2.$
So $C_i$ cannot be majority-approved to $C_{i+k}$.

Next, let $C_\ell$ be a candidate between $C_i$ and $C_{i+k}$. By the Observation,  $S_{i, \ell} \subset S_{i, i+k}$, and hence $\vert S_{i, \ell} \vert < N/2$. Thus, $C_i$ cannot be majority-preferred  to any of the candidates $C_{i+1}, C_{i+2}, \ldots, C_{i+k}.$

 Lastly, if $C_{i^\prime}$ is to the left of $C_i$, then again by the Observation,   $S_{i^\prime, i+k} \subset S_{i, i+k}$, which implies  $\vert S_{i^\prime, i+k} \vert< N/2$, and  if $C_\ell$ is between $C_i$ and $C_{i+k}$, then $S_{i^\prime, \ell} \subset S_{i, \ell}$ which implies $\vert S_{i^\prime, \ell} \vert< N/2$ . Thus $C_{i^\prime}$ also cannot be majority-preferred to any of the  candidates $C_{i+1}, C_{i+2}, \ldots, C_{i+k}$.
 
\end{proof}

From Proposition \ref{k_away}, we immediately get the following result, that restricts which of the extreme left and right candidates can beat more centrally-placed winning candidates in a head-to-head contest.  

\begin{proposition}\label{each_edge}
Suppose the winning set is $C_{i_1}, C_{i_2}, \ldots, C_{i_k}$. If $i_1\le k+1$ then no candidate on the left of $C_{i_1}$ can win in a head-to-head contest with any winning candidate. If $i_k\ge m-k$  then no candidate on the right of $C_{i_k}$ can win in a head-to-head contest with any winning candidate. 
\end{proposition}
\begin{proof}
Since $i_1-1 \le k$ and candidate $C_{i_1-1}$ is non-winning while candidate $C_{i_k}$ is winning, Proposition \ref{k_away} implies neither $C_{i_1-1}$ nor any candidate to their left can win in a head-to-head contest against any of the winning candidates. Similarly, since $i_k+1\ge m-k+1$ and $C_{i_k+1}$ is non-winning,  while $C_{i_1}$ is winning,  Proposition \ref{k_away}  implies  neither $C_{i_k+1}$ nor any candidate to their right can win in a head-to-head contest against any of the winning candidates.
\end{proof}

If the winning set is adjacent and $k$ is large enough,  Proposition \ref{each_edge} implies that it is also Gehrlein-stable.   

\begin{corollary}\label{edges} Suppose $k \ge \lceil \frac{m}{3} \rceil$. If the winning set, $C_i, C_{i+1}, \ldots, C_{i+k-1}$, is adjacent and satisfies $m-2k+1 \le i \le k+1$ then the winning set is Gehrlein-stable. 
\end{corollary}
(Note that  condition  $m-2k+1 \le i \le k+1$ can only  hold in  Corollary \ref{edges} if $m-2k+1 \le k+1$ or $k \ge \lceil \frac{m}{3} \rceil.$)

Together, Proposition \ref{adj} and Corollary \ref{edges} imply the following. 
\begin{proposition}\label{over_half}
If  $k \ge \lceil m/2 \rceil$ then the winning set is Gehrlein-stable. 
\end{proposition}
\begin{proof}
By Proposition \ref{adj}, the winning set must be adjacent. Suppose it is given by $C_i, C_{i+1}, \ldots, C_{i+k-1}$. Suppose, in addition, that $m$ is even. Since $k \ge m/2$, the furthest right position that $C_i$ can occupy $m/2+1 =k+1$; hence $i \le k+1$. Similarly, the farthest left position that $C_{i+k-1}$ can occupy is $m/2 \ge m-k$; hence $i+k-1 \ge m-k$. By Corollary \ref{edges}, the winning set is Gehrlein-stable. A similar argument applies when $m$ is odd.
\end{proof}

If $k< \lceil m/2 \rceil,$ then a winning set may fail to be adjacent. Even if it is adjacent,  the winning set may not be Gehrlein-stable. For instance, if $m=7$ , Corollary \ref{edges} implies that if $k=3$, the winning sets $C_2C_3C_4$, $C_3C_4C_5$  and $C_4C_5C_6$ are Gehrein-stable, however sets $C_1C_2C_3$ and $C_5C_6C_7$ may not be. 

Of course, if the winning set is not adjacent, it is possible that non-winning candidates lying between winning candidates may be majority-preferred to one or more winning candidates. Thus, no other winning sets are guaranteed to be Gehrlein-stable.

\subsection{Local Stability}

Since Gehrelin-stability implies local stability, Corollary \ref{edges} and Proposition \ref{over_half} also identify a number of winning sets that are locally stable.  But there are also instances when a winning set is not Gehrlein-stable and yet satisfies locally stability. This was the case in Proposition \ref{BD_set}, in which  $m=5$, $k=2$ and the winning set was $BD$ locally stable.

 Proposition \ref{each_edge} shows that candidates on the extreme left and right of a winning set are not majority preferred to winning candidates. However, if $k$ is relatively small in comparison to $m$, there may be non-winning candidates to the left of the left-most winner or to the right of the right-most winner that are majority-preferred to some or all of the winners. If such a candidate is majority-preferred to all the winners, the winning set will not be locally stable. This is the case, for instance, if the winning set lies entirely to the left or right of the Condorcet winner; this set will be neither a Condorcet set nor locally stable.
 
Finally, we consider the situation when there are non-winning candidates lying between winning candidates.   In this case, the non-winning candidate may be majority-preferred to each of the winners. However, as
 in Proposition \ref{BD_set},  if the gap between winners is small enough, the winning set will still be locally stable.

\begin{proposition}\label{gaps}
Suppose there is a gap between winning candidates of size $\ell \le k$ (so there are $\ell$ non-winning candidates between two winners).  If $C_{\ell}$ is a non-winning candidate in the gap, then the number of voters who prefer $C_{\ell}$ to all winning candidates is less than $N/2$. 
\end{proposition}
\begin{proof}
Let $T$ be the set of voters who prefer $C_{\ell}$ to every winning candidate. Since $\ell \le k$, $T \subset S_{\ell}.$  Let $C_{i_1}$ and $C_{i_k}$ be the left-most and right-most winning candidates. Since the candidates are more than $k$ apart,  $S_{i_1}$ and $S_{i_k}$ are disjoint. Hence we can partition  $S_{\ell}$ into three sets $S_{\ell}=T_1 \cup T_2 \cup T_3,$ where $T_1=S_{\ell} \cup S_{i_1}$ is the set of voters  who include  $C_{\ell}$ and $C_{i_1}$ but not $C_{i_k}$  among their top $k$ candidates, $T_2=S_{\ell} \cup S_{i_k}$ is the set of voters  who include $C_{\ell}$ and $C_{i_k}$ but not $C_{i_1}$ among their top $k$ candidates and $T_3=S_{\ell} \cup (S_{i_1} \cup S_{i_k})^c$ is the set of voters  who include  $C_{\ell}$ but neither  $C_{i_1}$ nor $C_{i_k}$  among their top $k$ candidates.  Then $\vert S_{\ell} \vert = \vert T_1 \vert+ \vert T_2 \vert+\vert T_3 \vert.$

 Next, let $S_{i_1} = L \cup T_1 $ where $L$ is the set of voters who  include $C_{i_1}$ but not $C_{\ell}$ among their top $k$, and let $S_{j_k} = R \cup T_2 $ where $R$ is the set of voters who  include $C_{i_k}$ but not $C_{\ell}$ among their top $k$.
 Then
 $$\vert S_{i_1} \vert = \vert L \vert+ \vert T_1 \vert \quad \text{and} \quad \vert S_{i_k} \vert = \vert R \vert+ \vert T_2 \vert.$$

 Since $C_{\ell}$ is a non-winning candidate and $C_{i_1}$ and $C_{i_k}$ are winning candidates,

 $$\vert T_1 \vert+ \vert T_2 \vert+\vert T_3 \vert = \vert S_{\ell} \vert < \vert S_{i_1} \vert = \vert L \vert+ \vert T_1 \vert,$$
which implies
 \begin{equation}\label{T_23}
 \vert T_2 \vert+\vert T_3 \vert < \vert L \vert,\end{equation}
and 
 $$ \vert T_1 \vert+ \vert T_2 \vert+\vert T_3 \vert \le \vert S_{\ell} \vert < \vert S_{i_k} \vert = \vert R \vert+ \vert T_2 \vert,$$
which implies
 \begin{equation}\label{T_13}
 \vert T_1 \vert+\vert T_3 \vert < \vert R \vert.\end{equation}

 Adding (\ref{T_23}) and (\ref{T_13}) together, we get

 $$\vert T_1 \vert+\vert T_2 \vert+2 \vert T_3 \vert< \vert L \vert+\vert R \vert,$$
which implies $\vert T \vert \le \vert S_{\ell} \vert < \vert L \vert+\vert R \vert.$
But $L$ and $R$ are disjoint. Moreover, if a voter is in $L$ or $R$, they are not in $T$; hence $L \cup R \subset T^c.$
So
$$\vert T \vert< \vert L \vert+\vert R \vert= \vert L \cup R \vert \le \vert T^c \vert.$$
Thus, $\vert T \vert < N/2$, as required.  
\end{proof}

Proposition \ref{gaps} shows that if the winning set is non-adjacent, then violations of local stability can  come   from non-winning candidates lying between winners only if the winning candidates are sufficiently far apart. Since a Condorcet winner will be majority-preferred to all other candidates, this implies that if the winning set does not include the Condorcet winner, the set will not be Condorcet, however it may be locally stable since different sets of voters will prefer the Condorcet winner to winning candidates on the left and on the right.

In   Sections \ref{sub_6} and \ref{sub_7}, we illustrate how the results in this Section characterize the properties of winning sets for $m=6$ and $m=7$ candidates. 

\subsection{Six Candidates}\label{sub_6}

Suppose there are 6 candidates, $A, B, C, D, E, F$ arranged left-to-right. By Propositions \ref{adj}  and \ref{over_half}, if $k \ge 3$, the winning sets must be adjacent and Gehrlein-stable. If $k=3$, these are $ABC$, $BCD$, $CDE$, and $DEF$.  If $k=4$, these are $ABCD$, $BCDE$,  and $CDEF.$ If $k=5$, these are $ABCDE$, and $BCDEF.$

If $k = 2$  the  winning sets  may be adjacent or non-adjacent. By Corollary \ref{edges},  the set $CD$ is  Gerlein-stable. Each  of remaining adjacent sets,  $AB, BC, CD, DE$ and $EF$,  winning set may not contain the Condorcet winner and hence may not be Condorcet or Gehrlein-stable. We consider them in turn.

The winning sets $AB$ and $EF$, which lie  on one extreme or the other of the spectrum fare the worst in head-to-head matches. In fact,  it is possible for every non-winning candidate to be majority-preferred to every winning candidate, as illustrated in Example \ref{ex_6_AB}. 

\begin{example}\label{ex_6_AB}
The winning set for the profile in Table \ref{AB_set} is $AB$, but every non-winning candidate is majority-preferred to every  winning candidate. Note that the winning set is neither Condorcet, since it does not contain the Condorcet winner, $D$, or locally stable, since the 5 right-most  voters prefer each of  the non-winning candidates to either of  the winners.

\begin{table}[h]\begin{center}
$$\begin{array}{ccc}
4 & 2 & 3 \\
\hline
A & D & F \\
B & C & E\\
C&  E & D \\
D & F & C \\
E & B & B\\
F & A & A
\end{array}$$
\end{center}
\caption{An example where the winning set is $AB$ and each non-winning  candidate is majority-preferred to each winning candidate.}
\label{AB_set}
\end{table}
\end{example}

If the winning set is $BC$ or $DE$, the winners that are closer to the center will not lose all their head-to-head contests. For example, if the winning set is $BC$ then  Proposition \ref{each_edge} implies $A$ cannot be majority-preferred to either winning candidate, and that $E$ and $F$ 
cannot be majority-preferred to $C$. However, it is possible for $E$ and $F$ to be majority-preferred to $B$, and for $D$ to be majority-preferred to both winning candidates as in the next example.

\begin{example} The winning set for the profile in Table \ref{BC_set} is $BC$,  However, a majority  of voters prefer $E$ or $F$ to $B$, and the same majority prefer $D$ (the Condorcet winner) to both winning candidates.  This winning set is neither Condorcet nor locally stable.
\begin{table}
\begin{center}
$$\begin{array}{ccc}
4 & 2 & 3 \\
\hline
C & D & F \\
B & C & E\\
A&  E & D \\
D & F & C \\
E & B & B\\
F & A & A
\end{array}$$
\end{center}
\caption{An example where the winning set is $BC$ is not Condorcet or locally stable.  }
\label{BC_set}
\end{table}
\end{example}

Lastly, we consider the non-adjacent winning sets: $BD, BE$ and $CE.$ By Proposition \ref{each_edge}, no candidate to the left of the left-most winner or to the right of the right-most winner is majority-preferred to all the winners. In addition, by Proposition \ref{gaps}, no candidate lying between winners can be preferred by the same majority to all the winners. Thus, each of these winning sets is locally stable. However, none of these sets is guaranteed to be Gehrlein-stable.  We leave it to the reader to construct profiles that satisfy the following.
\begin{itemize}
\item The winning set is $BD$,   non-winning candidates $C, E$ and $F$ are majority-preferred to $B$, and $C$ is  majority-preferred to $D$.
\item The winning set is $BE$, non-winning candidate $F$ is majority-preferred to $B$, and  $C$ and $D$ are majority-preferred to $B$ and $E.$
\end{itemize}

In summary, when $m=6$, we have the following.

\begin{proposition} \label{summary_6} Suppose there are $6$ candidates and $k = 2, \ldots, 5$. Then all winning sets except (if $k=2$), $AB, BC, DE$ and $EF$ are locally stable. If $k \ge 3$ then the winning sets are also adjacent and Gehrlein-stable. 
\end{proposition}

\subsection{Seven Candidates}\label{sub_7}

Suppose there are seven candidates, $A, B, C, D, E, F$ and $G$, arranged left to right. If $k \ge 4$, then by Propositions \ref{adj} and \ref{over_half}, each winning coalition is adjacent and Gehrlein-stable. 
If $k=3$ and the  winning set is $BCD$, $CDE$, or $DEF$ then Proposition \ref{k_away} implies that no losing candidate can win against a winning candidate; hence these are Gehrlein-stable. 

If the winning set is $ABC$ then neither $E, F$, nor $G$ can win against $B$ or $C$; however, all other head-to-head contests between winning and non-winning candidates could be won by the non-winning candidate, as demonstrated by the following example. (Similar statements hold for winning set $EFG$).

\begin{example}\label{ex_6_BE} The Bloc winning set for the  profile in Table \ref{ABC_winning} is $ABC$.  All non-winning candidates are majority-approved to $A$; in addition, $D$ (the Condorcet winner) is majority-approved by the five right-most voters to all winners, demonstrating that this winning set is not locally stable either.  

\begin{table}[h]
$$\begin{array}{ccc}
4 & 2 & 3\\
\hline
A & D & E \\
B & C & F \\
C&  B & G \\
D & E & D \\
E & F & C\\
F & G & B\\
G & A & A
\end{array}$$
\caption{A profile with winning set  $ABC$ that is not locally stable.}
\label{ABC_winning}
\end{table}
\end{example}

Next, we consider the non-adjacent winning sets,  $BCE$ and $CEF$. If the winning set is $BCE$, then by Proposition \ref{k_away}, $A$ cannot be majority-approved against to of the winning candidates, and  $F$ and $G$ cannot be majority-approved to $C$ or $E.$ But all other head-to-head contests between winning and non-winning candidates can be won by the non-winner. Bty Proposition \ref{gaps}, however, the winning coalition is locally stable. (Analogous statements apply if the winning set is $CEF$.)

\begin{example} The Bloc winning set for the  profile in Table \ref{BCE_winning} is  $BCE$. Non-winning candidates  $F$ and $G$ are majority-preferred by the five right-most candidates to $B$, and $D$ is majority-preferred to both $B$ and $C$. ($D$ is not preferred to $E$, who is the Condorcet winner.)  

\begin{table}[h]
\begin{center}
$$\begin{array}{cccc}
3 & 1&2 & 3\\
\hline
A & D&E & E \\
B & C&D & F \\
C& B& C & G \\
D &E& F & D \\
E &F& G & C\\
F & G&B & B\\
G & A&A & A
\end{array}$$
\end{center}
\caption{A profile with winning set  $BCE$ that is not Gehrlein-stable but is Condorcet and locally stable.}
\label{BCE_winning}
\end{table}
\end{example}

If $k=2$, the winning coalitions include: the adjacent sets $AB, BC, CD, DE, EF$ and $FG$; and the non-adjacent sets $BD, BE, BF, CE, CF$ and $DF$. Since $k$ is small relative to $m$, even centrally-placed winning sets such as $CD$ case, may not be Gehrlein-stable or locally stable.
Again, we leave it to the reader to find a profile such that the winning set is $CD$, the Condorcet winner is $F$, and all three non-winning candidates are majority-preferred to $D$. 

We summarize these remarks in the following Proposition.

\begin{proposition}\label{summary_7}
If $m=7$ and $k=3, 4$ or $5$, then all  winning sets are Gehrlein-stable with the exception (if $k=3$) of $ABC$ and $EFG$, $BCE$ and $CEF$. The sets $ABC$ and $EFG$ are not locally stable; the set $BCE$ and $CEF$ are locally stable. If $k=2$, none of the winning sets is guaranteed to be  Gehrlein-stable, Condorcet, or locally-stable.
\end{proposition}

\section{Simulations}
Propositions \ref{Geh-stable}, \ref{BD_set}, \ref{summary_6} and \ref{summary_7} identify which winning sets are guaranteed to be Gehrlein stable, Condorcet or locally-stable when $m=4, 5, 6$ or $7$. However, they say nothing about how {\it likely} each of these winning sets is to occur, and thus nothing about how likely the winning sets are to have any of these properties.

In this section, we investigate these questions using Monte Carlo simulations, which allow us to determine the likelihood of  each election outcome  based on specific assumptions about voter preferences.  Monte Carlo simulations have been used to investigate how often different voting methods agree \cite{MM24}, and how often they exhibit different undesirable properties \cite{MW23}, \cite{MW25}.  We  use three simulation models to analyze a large number of randomly generated elections with $m=4, 5, 6$ and $7$ candidates $k=2, \ldots, m-1$ winners. In each case, we compare the distribution of winning sets under Bloc  and $k$-Copland voting, and analyze how often the two voting methods agree. We also determine lower bounds on the probability of Bloc voting returning winning sets that are not  Gehrlein-stable or locally stable.

\subsection{Models}
We  use three different models of voter behavior.  In each case, we assume there are $10,001$ voters. (Values of $1001$ and $20001$ voters  did not materially effect the outcomes.)  Each model was run 100,000 times.

The first model, the Independent Anonymous Culture (IAC) model, samples randomly and uniformly from among all possible preference profiles of a fixed size. The IAC model has been widely used in empirical studies of voting methods because of its simplicity  (see  \cite{L93}, for example).
We adapt the IAC model to single-peaked preferences by restricting rankings to those that are compatible with the single-peaked setting. Such an approach has been used previously; see \cite{LCB96}, \cite{G03} and  \cite{MW23}, among others.
For example,  if there are  $5$ candidates then, using the notation of   Section \ref{sec_4_and_5}, there are 16 possible rankings, and each preference profile must satisfy
$$x_1 + x_2 + \cdots + x_{16} = 10,001.$$
There are $\binom{10,016}{15}$ sets of nonnegative integer solutions to this equation, each corresponding to a different preference profile.  The IAC model selects  randomly and uniformly from among these  profiles.

The other two models are spatial and based on a Euclidean norm (see \cite{KGF20}, and \cite{MW23} for previous analyses using spatial models.) Our first spatial model assumes voters lie along a number line according to a standard normal distribution; the second model assumes voter positions are bimodal, to capture the  effect of increased polarization among the electorate.

 More specifically, in the Euclidean Normal (EN) model,  voters are randomly assigned an ``ideal point''   on the real number line using a standard normal distribution with mean $0$ and standard deviation $1$. Candidate positions are selected  randomly (and uniformly) from among the voter positions.  Voters'  preferences are based on their Euclidean distance from each candidate: the candidate that is closest is ranked first, the candidate that is next closest is ranked second, and so on. Aggregating these rankings together yields a voter profile. The second model, the Euclidean Bimodal (EB)  is the same as the EN model except that  voter are assigned positions based on a bimodal distribution which is an equally weighted combination of  standard normal distributions centered at $-1$ and $1.$

\subsection{Results of the Simulations}
Results from the IAC simulations suggest that Bloc voting produces results that are very similar to those of the Copland method. Table \ref{IAC} lists the probabilities of each winning coalition occurring, as well as the probability of the two methods agreeing. For instance, with $m=5$ candidates and winning coalitions of size $k=2$, the probability of Bloc and Copland methods agreeing (and hence Bloc voting producing winning coalitions that are Gehrlein-stable), is $0.99420$

When profiles are unrestricted (not confined to a spatial model), the IAC model is known to produce   elections in which all candidates are closely tied. In the single-peaked setting when not all rankings are allowed, candidates towards the extreme left and right will be higher-ranked in fewer ballots and thus will occur noticeable less often than more centrist candidates. This is visible in Table \ref{IAC}  where, for example, when $m=6$ and $k=2$,  the winning coalitions $AB$ and $EF$, while theoretically possible,  occur with probability less than 0.00001. The same is true when $m=7$ and $k=2, 3$ or $4$ for winning coalitions such as $ABC$.
As a consequence, the probability under IAC that Bloc voting produces a wining coalition that is Gehrlein-stable when there are 7 or fewer candidates is extremely high.

\begin{table}[h]\label{IAC}
\begin{tabular}{ccc|cc}
Winning  & Bloc   & Copland& Winning Set & Bloc \&Copland  \\
Coalition & Probability & Probability&&Probability  \\
\hline
 \multicolumn{3}{c}{$m=4$ and $k=2$}   &&\\
 \hline
 $AB^*$ & 0.06200 & 0.06200 &\\
 $BC^*$ & 0.87596 & 0.87596 &\\
 $CD^*$ & 0.06204 & 0.06204 &\\
 \hline
 Agreement: &\multicolumn{2}{c}{1}   && \\
 \hline
\hline
 \multicolumn{3}{c}{$m=5$ and $k=2$}  & \multicolumn{2}{c}{$m=5$ and $k=3$}\\
\hline
$AB$ & 0.00126  & 0.00040& $ABC^*$ &0.01766 \\
$BC^*$ & 0.49684  & 0.49972& $BCD^*$ &0.96518  \\
$BD^\dagger$ & 0.00367  &0& $CDE^*$ &0.01716 \\
$CD^*$ & 0.49660 & 0.49952&&\\
$DE$ & 0.00142  &0.00036 &&\\
\hline
Agreement: &\multicolumn{2}{c}{0.99420}   & \multicolumn{2}{c}{1}\\
\hline
\hline
 \multicolumn{3}{c}{$m=6$ and $k=2$} & \multicolumn{2}{c}{$m=6$ and $k=3$ and $k=4$}\\
 \hline
$BC$ & 0.03010 &   0.01483 & $BCD^*$&0.49739 \\
$BD^\dagger$ & 0.00065  &  0 & $CDE^*$&0.50261 \\
\cline{4-5} 
$CD^*$ & 0.93785 &  0. 97008 &  $ABCD^*$ &0.00173 \\
$CE^\dagger$ & 0.00069 &  0& $BCDE^*$  &0.996661\\
$DE$ & 0.033071 &  0.01509& $CDEF^*$ &0.00166 \\
\hline
Agreement: &\multicolumn{2}{c}{0.96762}   & \multicolumn{2}{c}{1}\\
\hline
\hline
\multicolumn{3}{c}{$m=7$ and $k=2$ and $k=3$} & \multicolumn{2}{c}{$m=7$ and $k=4$ and $k=5$}\\
 \hline
$BC$ & 0.00003 &  0 & $BCDE^*$&0.49968 \\
$BD$ & 0.00001  &  0 & $CDEF^*$&0.50032 \\
$CD$ & 0.49865  &  0.50044 &&  \\
$DE$ & 0.50029  &  0.49956 && \\
$EF$ & 0.00004 &  0&  & \\
\hline
Agreement: &\multicolumn{2}{c}{0.92484}  &  \multicolumn{2}{c}{1}\\
\hline
$BCD^*$ & 0.00134& 0.00134& $ABCDE^*$ &0.00001  \\
$CDE^*$ &0.99745 &0.99745 & $BCDEF^*$ & 0.99997\\
$DEF^*$ & 0.00121& 0.00121& $CDEFG^*$ & 0.00002\\
\hline
Agreement: &\multicolumn{2}{c}{1}   & \multicolumn{2}{c}{1}\\
\hline
\hline
\hline
\end{tabular}
\caption{Probability of obtaining different winning coalitions using Bloc and $k$-Copland voting under the IAC model. Winning coalitions guaranteed to be Gehrlein-stable  denoted with a $*$; winning coalitions guaranteed to be locally (but not Gerhlein)-stable  denoted with a $\dagger$}
\label{IAC}
\end{table}

In contrast, the probability that Bloc and Copland methods agree under the spatial models  is much smaller. The results of the simulations under EN and EB models are  shown in Tables   \ref{EN_and_bimodal} and \ref{EN_and_bimodal_7}. Under both of these models, elections tend to be much less close, and hence the range of possible winning coalitions is much wider. This is particularly true under the EB model, where almost every possible winning coalition occurs (albeit with small probability).

\begin{table}[!ht]
\begin{tabular}{c|c|c|c|c}
\hline
 & \multicolumn{2}{c}{EN model}& \multicolumn{2}{c}{EB model}\\
\hline
$m/k$  & Gehrlein-stable   &  Locally-stable & Gehrlein-stable   & Locally-stable\\
\hline
$m=5$, $k=2$  &0.37985 &0.58014 & 0.03214 &0.45045\\
$m=5$, $k=3$  &1 & 1 & 1&1\\
$m=6$, $k=2$& 0.07518& 0.29785 &0.00063& 0.37809\\
$m=6$, $k=3,4$  &1 & 1 & 1&1\\
$m=7$, $k=2$&  0 &0 &0&0\\
$m=7$, $k=3$& 0.52902 &0.78172&0.04129&0.70927\\
$m=7$, $k=4, 5$  &1 & 1 & 1&1\\
\hline
\end{tabular}
\caption{Lower Bounds for the probability that winning coalitions under Bloc voting are Gehrlein-stable or locally-stable under EN and EB models}
\label{EN_and_bimodal_stable}
\end{table}

\begin{table}[h]
\begin{tabular}{c|cc|cc}
&\multicolumn{2}{c}{Standard Normal}  & \multicolumn{2}{c}{Bimodal Normal}\\
\hline
Winning  & Bloc   & Copland&  Bloc    & Copland \\
Set & Probability & Probability&Probability & Probability \\
\hline
\multicolumn{5}{c}{$m=4$ and $k=2$}  \\
\hline
$AB^*$ & 0.25000& 0.25000 & 0.25095& 0.25095 \\
$BC^*$ & 0.50097& 0.50097 &0.49879 &0.49879 \\
$CD^*$& 0.24903& 0.24903 & 0.25026& 0.25026\\
\hline
Agree: &\multicolumn{2}{c}{1}   & \multicolumn{2}{c}{1} \\
\hline
\hline
 \multicolumn{5}{c}{$m=5$ and $k=2$}  \\
\hline
$AB$ & 0.20922  & 0.12685  &0.27432 & 0.12435\\
$AD$&0 & 0& 0.00001 & 0 \\
$BC^*$ & 0.18870  & 0.37067  &0.01635 & 0.37682\\
$BD^\dagger$ & 0.20029  &0& 0.41831 &0\\
$CD^*$ & 0.19115 & 0.37700 &0.01578 & 0.37395\\
$DE$ & 0.21064  &0.12548  &0.27523 &0.12488\\
\hline
Agree: &\multicolumn{2}{c}{0.63218}   & \multicolumn{2}{c}{0.28136} \\
\hline
\hline
  \multicolumn{5}{c}{$m=5$ and $k=3$}  \\
 \hline
$ABC^*$& 0.25060 & 0.25060 & 0.24913 &0.24913  \\
$BCD^*$ & 0.50089 & 0.50089 &0.50252  &0.50252 \\
$CDE^*$ & 0.24851  &  0.24851 & 0.24835 & 0.24835\\
\hline
Agree: &\multicolumn{2}{c}{1}   & \multicolumn{2}{c}{1}  \\ 
\hline
\hline
  \multicolumn{5}{c}{$m=6$ and $k=2$}  \\
 \hline
 $AB$ & 0.18031  &0.06332 & 0.29907 &0.06360\\
$BC$ & 0.17212 &  0.25105&   0.01216& 0.25107 \\
$BD^\dagger$ & 0.02638 &  0& 0.02023 &0\\
$BE^\dagger$ & 0.16973 &  0& 0.33805  &0 \\
$CD^*$ & 0.07518 & 0.37476 &0.00063	& 0.37378 \\
$CE^\dagger$ & 0.02656 & 0 &0.01918&0\\
$DE$ & 0.17036 & 0.24852 & 0.01209	& 0.24959\\
$EF$ & 0.17936 & 0.06235 &0.29859& 0.06196\\
\hline
Agree: &\multicolumn{2}{c}{0.42981}   & \multicolumn{2}{c}{0.14050} \\
\hline
\hline
\multicolumn{5}{c}{$m=6$ and $k=3$} \\
 \hline
$ABC^*$ & 0.12434  & 0.12434 &0.12587	& 0.12587 \\
$BCD^*$ & 0.37874 &  0.37874 &0.37436& 0.37436\\
$CDE^*$& 0.37209 & 0.37209 &0.37360& 0.37360 \\
$DEF^*$ & 0.12483 & 0.12483 & 0.12617& 0.12617\\
\hline
Agree: &\multicolumn{2}{c}{1}   & \multicolumn{2}{c}{1} \\
\hline
\hline
\multicolumn{5}{c}{$m=6$ and $k=4$} \\
 \hline
 $ABCD^*$ & 0.24799 & 0.24799&0.25098	& 0.25098 \\
$BCDE^*$ & 0.50105& 0.50105 &0.49857& 0.49857\\
$CDEF^*$ & 0.25096  &  0.25096 & 0.25045& 0.25045 \\
\hline
Agree: &\multicolumn{2}{c}{1}   & \multicolumn{2}{c}{1} \\
\hline
\end{tabular}
\caption{Probability of obtaining different winning coalitions using Bloc and $k$-Copland voting under the EN and EB models for $4, 5$ and $6$ candidates. Winning coalitions guaranteed to be Gehrlein-stable  denoted with a $*$; winning coalitions guaranteed to be locally (but not Gerhlein)-stable  denoted with a $\dagger$.}
\label{EN_and_bimodal}
\end{table}

Table \ref{EN_and_bimodal_stable} provides a lower bound on the probability that a winner coalition will be Gehrlein-stable or locally-stable under either models. The numbers in the table are derived by simply summing up the probabilities of obtaining winning coalitions that are guaranteed to be Gehrlein or locally stable (marked by a $^*$ or a $^\dagger$ respectively in Tables \ref{EN_and_bimodal} and \ref{EN_and_bimodal_7} respectively). They do not include, for instance, the probability, of obtaining an individual  winning coalition  $AB$ (when $m=6$) that is Gehrlein-stable. Hence, the numbers in Table \ref{EN_and_bimodal_stable} are  \emph{underestimates} of the true probability of obtaining Gehrlein or locally stable winning coalitions under either model.

Lastly, we note that it is interesting to compare the winning coalitions under the EN and EB models. Not surprisingly, under the EN model, winning coalitions are much more likely to be closer to the center, whereas under the EB model, winning coalitions are much more likely to include a left and a right wing candidate. For instance, when $m=5$ and $k=2$, the locally stable coalition $BD$ occurs 41\% of the time---twice as often as under the EN model. Similar statements are true for $m=6$ and $k=2$ where $BE$ occurs  33\% versus 17\% of the time, and when $m=7$ and $k=2$ where $BE, BF, CF$ occur with  high frequency. It is also true that a coalition of only left or only right wing candidates is also more likely under the EB model: indeed when $m=7$ and $k=2$, the ``extreme'' winning coalitions $AB$ or $FG$ occur roughly half the time. 
\begin{table}[h]
\begin{tabular}{c|cc|cc}
&\multicolumn{2}{c}{Standard Normal}  & \multicolumn{2}{c}{Bimodal Normal}\\
\hline
Winning  & Bloc   & Copland&  Bloc    & Copland \\
Set & Probability & Probability&Probability & Probability \\
\hline
\multicolumn{5}{c}{$m=7$ and $k=2$} \\
 \hline
 $AB$ & 0.15774 & 0.03114 &0.25759& 0.03180\\
$BC$ & 0.15063& 0.15665 &0.04807 & 0.15792\\
$BD$ & 0.03155  &  0 & 0.00877 &0\\
$BE$ & 0.02974& 0  &  0.07233 &0\\
$BF$ & 0.13287 & 0  &  0.21884 &0\\
$CD$ & 0.06298  & 0.31240 &  0.00060 &0.30903\\
$CE$ & 0.00304 & 0 &0.00358&0 \\
$CF$ & 0.02800 & 0 &0.07477&0\\
$DE$ & 0.06433 & 0.31235& 0.00049&0.31338\\
$DF$ &0.03310 &0 &0.00897&0\\
$EF$ &0.14894 & 0.15603 &0.04844&0.015710\\
$FG$ & 0.15708 & 0.03143&0.25755&0.03077\\
\hline 
Agree: &\multicolumn{2}{c}{0.31119}   & \multicolumn{2}{c}{0.07584} \\
\hline
\hline
\multicolumn{5}{c}{$m=7$ and $k=3$} \\
 \hline
$ABC$ & 0.10895 & 0.06271 & 0.14455& 0.06166\\
$BCD^*$ & 0.14181 & 24916&  0.01195 & 0.24888\\
$BCE^\dagger$ & 0.12462 & 0 &0.33400 &0 \\
$CDE^*$ & 0.24687 & 0.37470 &0.01730& 0.37644\\
$CEF^\dagger$ & 0.12808 & 0 &0.33399&0 \\
$DEF^*$ & 0.14034 & 0.25066 &0.01204 &0.25149\\
$EFG$ & 0.10933 & 0.06277 &0.14617& 0.06153\\
\hline
Agree: &\multicolumn{2}{c}{0.65450}   & \multicolumn{2}{c}{0.16448} \\
\hline
\hline
\multicolumn{5}{c}{$m=7$ and $k=4$} \\
 \hline
$ABCD^*$ & 0.12466 & 0.12466 & 0.12346&0.12346\\
$BCDE^*$ & 0.37343 & 0.37343 &0.37693& 0.37693\\
$CDEF^*$ & 0.37578 & 0.37578 &0.37585& 0.37585\\
$DEFG^*$ & 0.12613 & 0.12613 &0.12376&0.12376\\
\hline
Agree: &\multicolumn{2}{c}{1}   &\multicolumn{2}{c}{1} \\
\hline
\hline
\multicolumn{5}{c}{$m=7$ and $k=5$} \\
 \hline
$ABCDE^*$ & 0.24929 & 0.24929 &0.25268&0.25268\\
$BCDEF^*$ & 0.50204 & 0.50204 &0.49806&0.49806\\
$CDEFG^*$ & 0.24867 & 0.24867 &0.24926& 0.24926\\
\hline
Agree: &\multicolumn{2}{c}{1}   &\multicolumn{2}{c}{1} \\
\hline
\end{tabular}
\caption{Wining Sets under Bloc and $k$-Copland methods under the EN and EB models for $7$ candidates.Probability of obtaining different winning coalitions using Bloc and $k$-Copland voting under the EN and EB models for $7$ candidates. Winning coalitions guaranteed to be Gehrlein-stable  denoted with a $*$; winning coalitions guaranteed to be locally (but not Gerhlein)-stable  denoted with a $\dagger$.}
\label{EN_and_bimodal_7}
\end{table}

\section{Conclusion}
In this paper, we have characterized the winning coalitions that emerge under Bloc voting when voters have single-peaked preferences. We have also investigated when the winning sets are adjacent and characterized them with respect to head-to-head contests. If the number of winners is at least half the number of candidates, then the winning coalitions consist of adjacent sets of candidates.  In this case, the winning set is Gehrlein-stable, meaning that no non-winning candidate can beat a winning candidate in a head-to-head contest.

If the number of winners is less than half the number of candidates, the winning coalitions may be non-adjacent. Even if they are adjacent, they may fail Gehrlein-stability.  Nevertheless, Propositions \ref{each_edge} provides restrictions on which non-winning candidates can beat winning candidates. Additionally, Proposition \ref{gaps} demonstrates that even small winning coalitions may satisfy local stability.

The Monte Carlo simulations provide  bounds for the likelihood that winning sets under Bloc voting will be Gehrlein or locally stable. Under the IAC model, in particular, where elections are frequently close, the winning sets are highly likely to be locally stable. This is less true under either spatial model.

Future research might extend this analysis to provide additional characterizations of winning sets for an arbitrary numbers of candidates. As more tools become available to map existing election data onto a one-dimensional spectrum, it would also be interesting to investigate how well Bloc voting fares on this empirical data.

Bloc voting is often dismissed \cite{EF17} as insufficiently able to capture voters' true preferences. However this analysis suggests that if $k$ is relatively large compared to $m$, that Bloc voting may perform similarly to a Condorcet-consistent method when voter preferences are single-peaked. Bloc voting also has the advantage of being simple to understand and widely used. Further analysis of its additional properties seems warranted.

\section*{Acknowledgments} 

The authors gratefully acknowledge the Institute for Mathematics and Democracy and its founder, Professor Ismar Volić, for providing the intellectual environment and collaborative framework that made this research possible. The Institute's mission to advance the mathematical analysis of democratic processes and voting systems has been instrumental in bringing our research team together and fostering the interdisciplinary dialogue that shaped this work.


{\footnotesize
\begin{thebibliography}{00}
\bibitem{B48} D. Black, On the rationale of group decision-making, {\it J.  Pol. Econ.}, {\bf 56} (1948), 23--34.

\bibitem{HE17} H. Aziz, E. Elkind, P. Faliszewski, M. Lackner \& P. Skowron, The condorcet principle for multiwinner elections: from shortlisting to proportionality, {\it Proceedings of the Twenty-Sixth International Joint Conference on Artificial Intelligence (IJCAI-17)}, (2017), 84--90. \url{https://doi.org/10.24963/ijcai.2017/13}

\bibitem{EF17} 
E. Elkind, P. Faliszewski, J. F. Laslier, P. Skowron, A. Slinko,  \& N. Talmon, What Do Multiwinner Voting Rules Do? An Experiment Over the Two-Dimensional Euclidean Domain, {\it Proceedings of the AAAI Conference on Artificial Intelligence} {\bf 31} (2017).  \url{https://doi.org/10.1609/aaai.v31i1.10612}


\bibitem{EF17a}
E. Elkind, P. Faliszewski,  P. Skowron, A. Slinko,. Properties of multiwinner voting rules. {\it Soc. Choice  Welf.}, {\bf  48} (2017), 599--632. \url{ https://doi.org/10.1007/s00355-017-1026-z}

\bibitem{G85}
W.V. Gehrlein, The Condorcet criterion and committee selection, {\it Math. Social Sci. }, {\bf 10}, (1985), 199--209.

\bibitem{G03} 
W.V. Gehrlein,  Condorcet efficiency and proximity to single-peaked preferences [Paper presentation] in: {\it Third international conference on logic, game theory and social choice}, Siena, (2003). 

\bibitem{L93} 
D. Lepelley,  On the probability of electing the Condorcet loser, {\it Math. Social Sci.}, {\bf 25}, (1993), 105--116.

\bibitem{LCB96} 
D; Lepelley, F. Chantreuil,  \& S. Berg,  The likelihood of monotonicity paradoxes in run-off elections.  {\it Math. Social Sci.}, {\bf 31}, (1996). 133--146.

\bibitem{LLS}
D. Lepelley, A. Louichi \& H. Smaoui, On Ehrhart polynomials and probability calculations in voting theory, {\it  Soc. Choice  Welf.}, {\bf 30}, (2008), 363--383. \url{ https://doi.org/10.1007/s00355-007-0236-1}

\bibitem{KGF20}
D.M. Kilgour, J.  Gregoire, \& A.M. Foley, The prevalence and consequences of ballot truncation in ranked-choice elections, {\it Pub. Choice}, {\bf 84}, (2020), 197--218.

\bibitem{MG24}
D. McCune, A. Graham-Squire, Monotonicity anomalies in Scottish local government elections, {\it Soc. Choice  Welf.}, {\bf 63} (2024), 69--101. \url{https://doi.org/10.1007/s00355-024-01522-5}

\bibitem{MM24}
D. McCune, E. Martin, G. Latina, G. \& K. Simms, A comparison of sequential ranked-choice voting and single transferable vote, {\it. J Comp. Soc. Sci.}, {\bf 7}, (2024), 643--670. \url{https://doi.org/10.1007/s42001-024-00249-8}

\bibitem{MW23} 
D. McCune and J. Wilson,  Ranked-Choice Voting and the Spoiler Effect, {\it Pu Ch} {\bf196},  (2023), 19--50. \url{https://doi.org/10.1007/s11127-023-01050-3}

\bibitem{MW25}
D. McCune and J. Wilson,  The negative participation paradox in three‑candidate instant runoff elections, {\it Theory and Decision}, {\bf 98},  (2025), 537--559.
\url{https://doi.org/10.1007/s11238-025-10026-2}

\bibitem{R03}
T.C.Ratliff,  Some startling inconsistencies when electing committees, {\it Soc. Choice  Welf.}, {\bf 21}, (2023), 43--454.

\bibitem{SF19}
P. Skowron, P.  Faliszewski,  \& A. Slinko,  Axiomatic characterization of committee scoring rules {\it  J .Econ. Th.}, {\bf 180}, (2019), 244--273.
https://doi.org/10.1016/j.jet.2018.12.011

\end{thebibliography}}
 \end{document}